\documentclass[11pt]{article}

\usepackage{geometry}
\usepackage{amsmath}
\usepackage{amssymb}
\usepackage{amsthm}
\usepackage[english]{babel}
\usepackage[utf8]{inputenc}
\usepackage{caption}
\usepackage{float}
\usepackage{graphicx}
\usepackage{xfrac}
\usepackage{textgreek}
\usepackage{array}
\usepackage{mathtools}
\usepackage{authblk}

\newtheorem{theorem}{Theorem}[section]

\usepackage{lineno,hyperref}
\modulolinenumbers[5]

\usepackage{amsmath}
\makeatletter
\newcommand{\xdashrightarrow}[2][]{\ext@arrow 0359\rightarrowfill@@{#1}{#2}}
\def\rightarrowfill@@{\arrowfill@@\relax\relbar\rightarrow}
\def\arrowfill@@#1#2#3#4{%
  $\m@th\thickmuskip0mu\medmuskip\thickmuskip\thinmuskip\thickmuskip
   \relax#4#1
   \xleaders\hbox{$#4#2$}\hfill
   #3$%
}
\makeatother

\begin{document}

\title{A new approach to simulating stochastic delayed systems}
\author{F. Fatehi$^{\rm 1}$, Y.N. Kyrychko$^{\rm 2}$, K.B. Blyuss$^{\rm 2}$\footnote{Corresponding author: K.Blyuss@sussex.ac.uk}}

\affil{$^{\rm 1}$ Department of Mathematics, University of York, York, YO10 5DD, UK}

\affil{$^{\rm 2}$ Department of Mathematics, University of Sussex, Falmer, Brighton, BN1 9QH, UK}

\maketitle

\begin{abstract}
In this paper we present a new method for deriving It\^{o} stochastic delay differential equations (SDDEs) from delayed chemical master equations (DCMEs). Considering alternative formulations of SDDEs that can be derived from the same DCME, we prove that they are equivalent both in distribution, and in sample paths they produce. This allows us to formulate an algorithmic approach to deriving equivalent It\^{o} SDDEs with a smaller number of noise variables, which increases the computational speed of simulating stochastic delayed systems. The new method is illustrated on a simple model of two interacting species, and it shows excellent agreement with the results of direct stochastic simulations, while also demonstrating a much superior speed of performance.  
\end{abstract}

\section{Introduction}

Stochastic models have successfully been used to study the dynamics of numerous biological processes across various scales, from gene regulation \cite{swain02,oz02} and immunology \cite{chan98,rib99,yuan11} to epidemics \cite{allen08,britton10} and population ecology \cite{matis}. Some of the most common methodologies used to analyse stochastic effects in biological models are continuous-time Markov chains (CTMC), discrete-time Markov chains (DTMC), and stochastic differential equations (SDEs) \cite{allen2007b}. Focusing on continuous-time models, CTMC are formulated in terms of probabilities of transitions between different states under memoryless assumption, and they result in the forward Kolmogorov equation, also known as the chemical master equation (CME), which, with an exception of some very simple examples, cannot be solved analytically. To make further analytical progress, one can then either use the CME to derive a system of equations for moments of the distribution and use some higher-order approximation to make this a closed system of differential equations, or one can use approximations, such as van Kampen or Kramers-Moyal expansions \cite{van1992,gardiner}, to obtain Gaussian approximations for dynamics around deterministic trajectories. Alternatively, one can solve the CME numerically using, e.g., Gillespie's exact stochastic simulation algorithm (SSA) \cite{gillespie77} or some alternative formulations \cite{cao04,gibson00,anderson07}. Another approach is to use forward Kolmogorov equation to reformulate the problem as an SDE, which, for a large system size would provide a good approximation of the underlying CTMC dynamics \cite{eka10}. Although being only an approximation of the exact stochastic dynamics, this approach has a major advantage of being very computationally efficient, since numerical solutions of SDEs can be found at a fraction of time required for the full simulation of the original CTMC. A very recent paper by Warne et al. \cite{warne19} provides a nice overview of these and other different approaches to simulating biochemical reactions.

Besides stochasticity, many biological processes are also characterised by non-negligible time delays, such as, intracellular delays associated with gene transcription and translation \cite{Hirata02,jensen03}, latency and immunity periods in epidemics \cite{heth89,lloyd00}, or maturation period in ecology \cite{kuang}. Thus, it is essential to correctly account for those delays in corresponding mathematical models. Similar to CMEs for non-delayed models, one can analyse stochastic delayed systems using the {\it delay chemical master equation} (DCME) that describes the exact probability distribution of finding the system in a particular state \cite{Bratsun2005,Barrio2006,tian2007}. Leier and Marquez-Lago \cite{leier2015} have presented a general framework of DCMEs, which covers both consuming and non-consuming delayed reactions, and applies not only to fixed time delays, but also to delay distributions. They showed how one can obtain closed-form solutions of the DCME for some simple reaction schemes. Galla \cite{galla2009} showed how one can perform a system-size expansion of the DCME to obtain a delayed Langevin equation describing fluctuations around solutions of the deterministic models (see also Guillouzic et al. \cite{guillouzic1999} and Phillips et al. \cite{phillips2016} for further examples of using this approach). Brett and Galla \cite{brett_galla2013,brett_galla2014} showed how one can derive chemical Langevin equation describing deterministic limit and linear-noise approximation around it for stochastic models with distributed delay, without using a master equation, but instead relying on the generating functional approach.

In terms of numerical simulations of stochastic delayed models, one of the first approaches to modelling the combined effects of a time delay and intrinsic noise was proposed by Bratsun et al. \cite{Bratsun2005} in the context of gene regulation. They developed a truncated master equation for a set of biochemical reactions, some of which are delayed, and also introduced modifications to the Gillespie algorithm to incorporate delayed reactions. Barrio et al. \cite{Barrio2006} developed a delay stochastic simulation algorithm (DSSA) based on the so-called `rejection method', which accounts for waiting times and also provides a method for simulating {\it consuming delayed reactions}, defined as such reactions where the reactants of an unfinished reaction cannot participate in a new reaction. In this respect, the rejection method is superior to the algorithm of Bratsun et al. \cite{Bratsun2005}, which can only be used for simulating {\it non-consuming delayed reactions}, in which the reactants of an unfinished reaction can also participate in other reactions. Zavala and Marquez-Lago \cite{Zavala2014} have used rejection algorithm to study stochastic effects in a simple genetic circuit with negative feedback and transcriptional/translational delays. Subsequently, Cai \cite{Cai2007} developed a so-called `direct algorithm' and showed that this method, as well as the rejection method of Barrio et al. \cite{Barrio2006}, is exact, with the direct algorithm being faster and generating fewer additional random variables. More recently, Thanh et al. \cite{Thanh2014,Thanh2017} proposed some further DSSAs with improved computational performance. Marquez-Lago et al. \cite{lago10} developed a DSSA that can work not only with discrete delay, but also with delay distributions.

Since using DSSAs can be very computationally demanding \cite{phillips2016,Niu2015b}, one can use SDDEs that obtain an approximation for DCMEs in the same way as SDEs provide an approximation for CMEs, with the advantage of such approach being much more computationally efficient. Tian et al. \cite{tian2007} developed two methods for deriving SDDEs from discrete delayed stochastic models with non-consuming delayed reactions, and then used the Euler-Maruyama method for solving them for fixed time delay, as well as for time delay obeying a uniform distribution or being a Gaussian random variable. The results of simulations on a simple model of gene regulatory network showed small differences in means and variances between two SDDE models. As an alternative, Niu et al. \cite{Niu2015a,Niu2015b} have introduced a strong predictor-corrector method for numerical solution of SDDEs and showed that its asymptotic mean-square stability bound is much larger than that of the Euler-Maruyama method, while its implementation is much more efficient. Frank \cite{Frank2002} has shown how the probability distribution of a SDDE can be described analytically as a solution of a delayed Fokker-Planck equation (DFPE), and also proposed a method for deriving a DFPE directly from SDDEs \cite{Frank}. 

In this paper, we propose a method for deriving It\^{o} SDDEs from DCMEs for different types of delayed reactions, which generalises the methodology of Tian et al. \cite{tian2007} to also include consuming delayed reactions. We will adapt an approach used by Allen et al. \cite{allen2008} for non-delayed stochastic equations to prove that alternative forms of such SDDEs are equivalent both in distribution, and in sample path trajectories they produce, thus addressing the above-mentioned issue of small differences between numerical realisations of alternative SDDEs in Tian et al. \cite{tian2007}. This allows us to formulate an algorithmic approach for deriving a computationally efficient It\^{o} SDDE with a smaller number of noise variables. Using an example of a system with two interacting species that contains non-delayed and delayed reactions (both non-consuming, and consuming), we will illustrate the efficiency of our method in terms of computational speed and comparison with direct simulation using DSSA.

\section{It{\^o} SDDE models and their equivalence}

As a starting point, we consider a system of  $N$ molecular species (which can also represent cells, biological populations etc.) $S = \{S_1,\ldots, S_N\}$, whose state at time $t$ is described by a vector $\textbf{X}(t)=(X_1(t),\ldots,X_N(t))$, and these species react through reactions $\{R_1,\ldots,R_m\}$. Each reaction $R_j$ is characterised by a state change vector ${\bf v}_j=(v_{1j},v_{2j},\ldots,v_{Nj})^T$, and an associated propensity function $a_j$. One has to explicitly distinguish between non-consuming and consuming delayed reactions, because non-delayed and non-consuming delayed reactions have a single update vector ${\bf v}$, whereas for delayed consuming reactions, $\textbf{v}_j^r$ and $\textbf{v}_j^p$ are the update vectors for reactants at the start of reaction, and for products at the end of the time delay associated with reaction $R_j$, respectively, so $\textbf{v}_j^r+\textbf{v}_j^p=\textbf{v}_j$. Assuming the first $m_1$ reactions to be non-delayed, the reactions $m_1+1$ to $m_2$ to be delayed non-consuming reactions with corresponding time delays $\tau_{m_1+1},\ldots,\tau_{m_2}$, and the rest to be consuming delayed reactions with time delays $\tau_{m_2+1},\ldots,\tau_{m}$, the DCME accounting for all non-consuming and consuming reactions is then given by \cite{leier2015}
\begin{align}\label{DCME1}
\dfrac{\partial}{\partial t}&P(\textbf{X},t)=-\sum\limits_{j=1}^{m_1}a_j(\textbf{X})P(\textbf{X},t)+\sum\limits_{j=1}^{m_1}a_j(\textbf{X}-\textbf{v}_j)P(\textbf{X}-\textbf{v}_j,t) \nonumber\\
&-\sum\limits_{j=m_1+1}^{m_2}\sum\limits_{\textbf{X}_i\in I(\textbf{X})}a_j(\textbf{X}_i)P(\textbf{X},t;\textbf{X}_i,t-\tau_j)+\sum\limits_{j=m_1+1}^{m_2}\sum\limits_{\textbf{X}_i\in I(\textbf{X})}a_j(\textbf{X}_i)P(\textbf{X}-\textbf{v}_j,t;\textbf{X}_i,t-\tau_j) \nonumber\\
&-\sum\limits_{j=m_2+1}^{m}\sum\limits_{\textbf{X}_i\in I(\textbf{X})}a_j(\textbf{X}_i)P(\textbf{X},t;\textbf{X}_i,t-\tau_j)+\sum\limits_{j=m_2+1}^{m}\sum\limits_{\textbf{X}_i\in I(\textbf{X})}a_j(\textbf{X}_i)P(\textbf{X}-\textbf{v}_j^p,t;\textbf{X}_i,t-\tau_j) \nonumber\\
&-\sum\limits_{j=m_2+1}^{m}a_j(\textbf{X})P(\textbf{X},t)+\sum\limits_{j=m_2+1}^{m}a_j(\textbf{X}-\textbf{v}_j^r)P(\textbf{X}-\textbf{v}_j,t),
\end{align}
where $I(\textbf{X})$ is the set of all possible system states in the past, from which the given state $\textbf{X}$ can follow via a chain of reactions, and $P(\textbf{X},t;\textbf{X}_i,t-\tau_i)$ is the joint probability of finding the system in state $\textbf{X}$ at time $t$, and in state $\textbf{X}_i$ at time $t-\tau_i$, with $P(\textbf{X},t)=P(\textbf{X},t;\textbf{X}_0,t_0,\mathcal{H}_0)$, where $\mathcal{H}_0$ is initial history. Let $\textbf{Y}(t)=(Y_1(t),Y_2(t),\ldots,Y_N(t))^T$ be a vector of continuous random variables representing the amounts of molecular species at time $t$. Applying the methodology as used in Tian et al. \cite{tian2007} for systems without consuming delays, the corresponding SDDE model which faithfully represents the intrinsic noise associated with all those delayed reactions, has the form
\begin{align}\label{SDDE model}
d\textbf{Y}=&\sum\limits_{j=1}^{m_1}\textbf{v}_ja_j(\textbf{Y}(t))dt+\sum\limits_{j=m_1+1}^{m_2}\textbf{v}_ja_j(\textbf{Y}(t-\tau_j))dt \nonumber\\
&+\sum\limits_{j=m_2+1}^{m}\textbf{v}_j^ra_j(\textbf{Y}(t))dt+\sum\limits_{j=m_2+1}^{m}\textbf{v}_j^pa_j(\textbf{Y}(t-\tau_j))dt \nonumber\\
&+\sum\limits_{j=1}^{m_1}\textbf{v}_j\sqrt{a_j(\textbf{Y}(t))}dW_j(t)+\sum\limits_{j=m_1+1}^{m_2}\textbf{v}_j\sqrt{a_j(\textbf{Y}(t-\tau_j))}dW_j(t) \nonumber\\
&+\sum\limits_{j=m_2+1}^{m}\textbf{v}_j^r\sqrt{a_j(\textbf{Y}(t))}dW_j(t)+\sum\limits_{j=m_2+1}^{m}\textbf{v}_j^p\sqrt{a_j(\textbf{Y}(t-\tau_j))}dW_{j-m_2+m}(t)\nonumber\\
=&\textbf{f}\left(\textbf{Y}(t),\textbf{Y}(t-\tau_{m_1+1}),\dots,\textbf{Y}(t-\tau_m)\right)dt+H\left(\textbf{Y}(t),\textbf{Y}(t-\tau_{m_1+1}),\dots,\textbf{Y}(t-\tau_m),t\right)d\textbf{W}(t),
\end{align}
where $\textbf{W}(t)=(W_1(t),W_2(t),\ldots,W_{2m-m_2})^T$ is a vector of independent Wiener processes, and\\$H=\big(H_1\quad H_2\quad H_3\quad H_4\big)$ is a $N\times (2m-m_2)$ matrix which
\[
H_1=\begin{pmatrix}
v_{11}\sqrt{a_1(\textbf{Y}(t))} & v_{12}\sqrt{a_2(\textbf{Y}(t))} & \cdots & v_{1m_1}\sqrt{a_{m_1}(\textbf{Y}(t))} \\
v_{21}\sqrt{a_1(\textbf{Y}(t))} & v_{22}\sqrt{a_2(\textbf{Y}(t))} & \cdots & v_{2m_1}\sqrt{a_{m_1}(\textbf{Y}(t))} \\
\vdots & \vdots & \ddots & \vdots \\
v_{N1}\sqrt{a_1(\textbf{Y}(t))} & v_{N2}\sqrt{a_2(\textbf{Y}(t))} & \cdots & v_{Nm_1}\sqrt{a_{m_1}(\textbf{Y}(t))}
\end{pmatrix}_{N\times m_1},
\]\\
\[
H_2=\begin{pmatrix}
v_{1(m_1+1)}\sqrt{a_{m_1+1}(\textbf{Y}(t-\tau_{m_1+1}))} & \cdots & v_{1m_2}\sqrt{a_{m_2}(\textbf{Y}(t-\tau_{m_2}))} \\
v_{2(m_1+1)}\sqrt{a_{m_1+1}(\textbf{Y}(t-\tau_{m_1+1}))} & \cdots & v_{2m_2}\sqrt{a_{m_2}(\textbf{Y}(t-\tau_{m_2}))} \\
\vdots & \ddots & \vdots \\
v_{N(m_1+1)}\sqrt{a_{m_1+1}(\textbf{Y}(t-\tau_{m_1+1}))} & \cdots & v_{Nm_2}\sqrt{a_{m_2}(\textbf{Y}(t-\tau_{m_2}))} \\
\end{pmatrix}_{N\times (m_2-m_1)},
\]\\
\[
H_3=\begin{pmatrix}
v_{1(m_2+1)}^r\sqrt{a_{m_2+1}(\textbf{Y}(t))} & v_{1(m_2+2)}^r\sqrt{a_{m_2+2}(\textbf{Y}(t))} & \cdots & v_{1m}^r\sqrt{a_{m}(\textbf{Y}(t))} \\
v_{2(m_2+1)}^r\sqrt{a_{m_2+1}(\textbf{Y}(t))} & v_{2(m_2+2)}^r\sqrt{a_{m_2+2}(\textbf{Y}(t))} & \cdots & v_{2m}^r\sqrt{a_{m}(\textbf{Y}(t))} \\
\vdots & \vdots & \ddots & \vdots \\
v_{N(m_2+1)}^r\sqrt{a_{m_2+1}(\textbf{Y}(t))} & v_{N(m_2+2)}^r\sqrt{a_{m_2+2}(\textbf{Y}(t))} & \cdots & v_{Nm}^r\sqrt{a_{m}(\textbf{Y}(t))} \\
\end{pmatrix}_{N\times (m-m_2)},
\]\\
\[
H_4=\begin{pmatrix}
v_{1(m_2+1)}^p\sqrt{a_{m_2+1}(\textbf{Y}(t-\tau_{m_2+1}))} & \cdots & v_{1m}^p\sqrt{a_{m}(\textbf{Y}(t-\tau_{m}))} \\
v_{2(m_2+1)}^p\sqrt{a_{m_2+1}(\textbf{Y}(t-\tau_{m_2+1}))} & \cdots & v_{2m}^p\sqrt{a_{m}(\textbf{Y}(t-\tau_{m}))} \\
\vdots & \ddots & \vdots \\
v_{N(m_2+1)}^p\sqrt{a_{m_2+1}(\textbf{Y}(t-\tau_{m_2+1}))} & \cdots & v_{Nm}^p\sqrt{a_{m}(\textbf{Y}(t-\tau_{m}))} \\
\end{pmatrix}_{N\times (m-m_2)}.
\]
In this formulation, each delayed consuming reaction is effectively split into two reactions, one describing changes in reactants, and one describing changes in products, in the same way as they are represented in the DCME (\ref{DCME1}). This then results in extending the number of independent Wiener processes that need to be included in the SDDE (\ref{SDDE model}) in a manner similar to how non-delayed and delayed non-consuming reactions are treated. This also fits with an underlying assumption of weak coupling of the system states at time $t$ and $t-\tau_{k}$, which underlies the derivation of the SDDE (\ref{SDDE model}) from the DCME in \cite{tian2007}, and one should also note that a similar approach is taken when one performs system-size expansion of the DCME \cite{galla2009}.

Tian et al. \cite{tian2007} have also considered an alternative formulation of the model in the form
\begin{align}\label{SDDE2}
d\textbf{Y}^{\ast}(t)=&{\bf f}\left(\textbf{Y}^{\ast}(t),\textbf{Y}^{\ast}(t-\tau_{m_1+1}),\dots,\textbf{Y}^{\ast}(t-\tau_{m}),t\right)dt\nonumber\\
&+G\left(\textbf{Y}^{\ast}(t),\textbf{Y}^{\ast}(t-\tau_{m_1+1}),\dots,\textbf{Y}^{\ast}(t-\tau_m),t\right)d\textbf{W}^{\ast}(t),
\end{align}
where $\textbf{Y}^{\ast}(t)=(Y_1^{\ast}(t),Y_2^{\ast}(t),\dots,Y_N^{\ast}(t))^T$, 
$\textbf{W}^{\ast}(t)=(W_1^{\ast}(t),W_2^{\ast}(t),\dots,W_N^{\ast}(t))^T$, with $W_j^{\ast}$, $1\leq j\leq N$, being independent Wiener processes, and 
$G$ being an $N\times N$ symmetric positive semidefinite matrix related to $H$ through an $N\times N$ matrix $V$, where $V=HH^T$ and $G=V^{\sfrac{1}{2}}$, which also implies $V=GG^T$. As mentioned earlier, numerical simulations of a model for gene regulatory networks using these two alternative SDDE formulations produced small differences in observed means and variances of resulting distributions, thus is was suggested that ``{\it more work is needed to compare the difference between the two types of the Langevin approach}" \cite{tian2007}. To address this problem, we will now show that the above two SDDE models are actually equivalent in the sense that their solutions have the same probability distribution, as well as the same sample path solutions.

To show that systems (\ref{SDDE model}) and (\ref{SDDE2}) are equivalent in distribution, i.e. their solutions have the same probability distribution, it suffice to show that the probability density function for both of these systems satisfies the same forward Kolmogorov or Fokker-Planck equation. This is established by the following result, which generalises earlier work in \cite{Frank,Risken} to the case of multiple time delays and multi-dimensional stochastic systems.

\begin{theorem}\label{NewThm}
	Consider the following It\^{o} SDDE model
	\begin{align*}
	dY_i(t)=&f_i\left(\textbf{Y}(t),\textbf{Y}(t-\tau_1),\dots,\textbf{Y}(t-\tau_q),t\right)dt\\
	&+\sum_{j=1}^{M}g_{ij}\left(\textbf{Y}(t),\textbf{Y}(t-\tau_1),\dots,\textbf{Y}(t-\tau_q),t\right)dW_j(t),
	\end{align*}
	where $W_j(t)$ are independent Wiener processes, and
	$$f_i:\underbrace{\mathbb{R}^N\times \mathbb{R}^N\times \dots\times \mathbb{R}^N}_\text{(q+1)-times}\times \mathbb{R}\rightarrow \mathbb{R},\quad g_{ij}:\underbrace{\mathbb{R}^N\times \mathbb{R}^N\times \dots\times \mathbb{R}^N}_\text{(q+1)-times}\times \mathbb{R}\rightarrow \mathbb{R},$$
	for every $1\leq i\leq N$ and $1\leq j\leq M$, with the initial condition $\textbf{Y}(t)=\boldsymbol{\varphi}(t)$ for $t\in [-\tau,0]$, where $\tau=\max\{\tau_1,\dots,\tau_q\}$. 
	The corresponding delay Fokker-Planck equation has the form
	\begin{align*}
	\dfrac{\partial}{\partial t}&P(\textbf{y},t\mid\boldsymbol{\varphi})=\nonumber\\
	&-\sum\limits_{i=1}^N\dfrac{\partial}{\partial y_i}\underbrace{\idotsint}_\text{q}f_i(\textbf{y},\textbf{y}_{\tau_1},\dots,\textbf{y}_{\tau_q},t)P(\textbf{y},t;\textbf{y}_{\tau_1},t-\tau_1;\dots;\textbf{y}_{\tau_q},t-\tau_q\mid\boldsymbol{\varphi})dV\nonumber\\
	&+\dfrac{1}{2}\sum\limits_{i,j}\dfrac{\partial^{2}}{\partial y_i \partial y_j}\underbrace{\idotsint}_\text{q}\left(GG^T\right)_{ij}P(\textbf{y},t;\textbf{y}_{\tau_1},t-\tau_1;\dots;\textbf{y}_{\tau_q},t-\tau_q\mid\boldsymbol{\varphi})dV,
	\end{align*}
	where $dV=d{\textbf{y}_{\tau_1}} d{\textbf{y}_{\tau_2}}\dots d{\textbf{y}_{\tau_q}}$, and $G$ is an $N\times M$ matrix with $G_{ij}=g_{ij}(\textbf{y},\textbf{y}_{\tau_1},\dots,\textbf{y}_{\tau_q},t)$, for every $1\leq i\leq N$ and $1\leq j\leq M$.
\end{theorem}

\begin{proof}
	Let us consider the joint probability density
	\begin{equation*}
	\resizebox{1\hsize}{!}
	{
		$P(\textbf{y},t;\textbf{y}^{\prime},t^{\prime};\textbf{y}_{\tau_1},t^{\prime}-\tau_1;\dots;\textbf{y}_{\tau_q},t^{\prime}-\tau_q\mid\boldsymbol{\varphi})=\bigg \langle\delta\big(\textbf{y}-\textbf{y}(t)\big) \delta\big(\textbf{y}^{\prime}-\textbf{y}(t^{\prime})\big)\prod\limits_{k=1}^q\delta\big(\textbf{y}_{\tau_k}-\textbf{y}(t^{\prime}-\tau_k) \big)\bigg\rangle$
	}
	\end{equation*}
	for $t\geq t^{\prime}$, where $\langle\ldots\rangle$ denotes ensemble average, and $\delta(\cdot)$ is the Dirac delta function. Expressing the single time-point probability density $P(\textbf{y},t\mid\boldsymbol{\varphi})$ through the conditional probability density and utilising the generalized Kramers-Moyal expansion \cite{Frank,Risken} yields the following PDE
	\begin{equation*}
	\resizebox{1\hsize}{!}
	{
		$\dfrac{\partial}{\partial t}P(\textbf{y},t\mid\boldsymbol{\varphi})=\sum\limits_{\upsilon=1}^{\infty}\sum\limits_{j_1,j_2,\dots,j_\upsilon}\dfrac{(-\partial)^\upsilon}{\partial y_{j_1}\dots\partial y_{j_\upsilon}}\underbrace{\idotsint}_{q}D^{(\upsilon)}_{j_1\dots j_\upsilon}P(\textbf{y},t;\textbf{y}_{\tau_1},t-\tau_1;\dots;\textbf{y}_{\tau_q},t-\tau_q\mid\boldsymbol{\varphi})dV,$
	}
	\end{equation*}
	where $dV=d{\textbf{y}_{\tau_1}} d{\textbf{y}_{\tau_2}}\dots d{\textbf{y}_{\tau_q}}$, and $D^{(\upsilon)}_{j_1\dots j_\upsilon}(\cdot)$ are given by
	\begin{equation*}
	\begin{aligned}
	D^{(\upsilon)}_{j_1\dots j_\upsilon}(&\textbf{y},\textbf{y}_{\tau_1},\dots,\textbf{y}_{\tau_q},t)=\\
	&\lim\limits_{h\rightarrow 0}\dfrac{1}{h}\displaystyle\int\dfrac{\prod\limits_{k=1}^{\upsilon}(z_{j_k}-y_{j_k})}{\upsilon !}P(\textbf{z},t+h\mid\textbf{y},t;\textbf{y}_{\tau_1},t-\tau_1;\dots;\textbf{y}_{\tau_q},t-\tau_q;\boldsymbol{\varphi})d\textbf{z}.
	\end{aligned}
	\end{equation*}
	Since we are working with an It\^{o} SDDE, it is possible to reformulate the problem in the form of Langevin equation similar to the case of Markov process \cite{Frank}. By rewriting coefficients $\displaystyle{D^{(\upsilon)}_{j_1\dots j_\upsilon}}$ in the form
	\begin{equation}\label{KM coeff}
	D^{(\upsilon)}_{j_1\dots j_\upsilon}(\textbf{y},\textbf{y}_{\tau_1},\dots,\textbf{y}_{\tau_q},t)=
	\lim\limits_{h\rightarrow 0}\dfrac{1}{h}\left.\dfrac{\left\langle\prod\limits_{k=1}^{\upsilon}\big(Y_{j_k}(t+h)-Y_{j_k}(t)\big)\right\rangle}{\upsilon !}\right\vert_{\textbf{Y}(t)=\textbf{y},\textbf{Y}(t-\tau_1)=\textbf{y}_{\tau_1},\dots,\textbf{Y}(t-\tau_q)=\textbf{y}_{\tau_q}},
	\end{equation}
	one can use the time-discrete version of the SDDE model \cite{Frank,Risken} to obtain the following expressions for these coefficients
	\begin{equation*}
	\begin{aligned}
	&D^{(1)}_{i}(\textbf{y},\textbf{y}_{\tau_1},\dots,\textbf{y}_{\tau_q},t)=f_i(\textbf{y},\textbf{y}_{\tau_1},\dots,\textbf{y}_{\tau_q},t),\\
	&D^{(2)}_{ij}(\textbf{y},\textbf{y}_{\tau_1},\dots,\textbf{y}_{\tau_q},t)=\dfrac{1}{2}\sum\limits_{k=1}^m g_{ik}(\textbf{y},\textbf{y}_{\tau_1},\dots,\textbf{y}_{\tau_q},t)g_{jk}(\textbf{y},\textbf{y}_{\tau_1},\dots,\textbf{y}_{\tau_q},t),\\
	&D^{(\upsilon)}_{j_1\dots j_\upsilon}(\textbf{y},\textbf{y}_{\tau_1},\dots,\textbf{y}_{\tau_q},t)=0,\quad \mbox{for every }\upsilon\geq 3,
	\end{aligned}
	\end{equation*}
	which completes the proof.
\end{proof}

Due to the relation $V=GG^T=HH^T$, {\bf Theorem~\ref{NewThm}} implies that solutions to SDDEs (\ref{SDDE model}) and (\ref{SDDE2}) do indeed have the same probability distribution. We now use the method presented in Allen et al. \cite{allen2008} for non-delayed SDEs to show that sample paths obtained as solutions of one of these SDDEs are also sample paths of the other SDDE.

\begin{theorem}\label{NewThm2}
Consider the two following It\^{o} SDDE systems
\begin{align}\label{SDDE1}
d\textbf{Y}(t)=&\textbf{f}\left(\textbf{Y}(t),\textbf{Y}(t-\tau_1),\dots,\textbf{Y}(t-\tau_q),t\right)dt\nonumber\\
&+G\left(\textbf{Y}(t),\textbf{Y}(t-\tau_1),\dots,\textbf{Y}(t-\tau_q),t\right)d\textbf{W}(t),
\end{align}
and
\begin{align}\label{SDDE3}
d\textbf{Y}^{\ast}(t)=&\textbf{f}\left(\textbf{Y}^{\ast}(t),\textbf{Y}^{\ast}(t-\tau_1),\dots,\textbf{X}^{\ast}(t-\tau_q),t\right)dt\nonumber\\
&+B\left(\textbf{Y}^{\ast}(t),\textbf{Y}^{\ast}(t-\tau_1),\dots,\textbf{Y}^{\ast}(t-\tau_q),t\right)d\textbf{W}^{\ast}(t),
\end{align}
where $\textbf{Y}(t)=(Y_1(t),Y_2(t),\dots,Y_n(t))^T$, $\textbf{Y}^{\ast}(t)=(Y_1^{\ast}(t),Y_2^{\ast}(t),\dots,Y_n^{\ast}(t))^T$, $\textbf{W}(t)=(W_1(t),\allowbreak W_2(t),\dots,W_m(t))^T$, $\textbf{W}^{\ast}(t)=(W_1^{\ast}(t),W_2^{\ast}(t),\dots,W_n^{\ast}(t))^T$, and matrices $G$ and $B$ are related through the $n\times n$ matrix $V$, where $V=GG^T$ and $B=V^{\sfrac{1}{2}}$. Notice that $V$ and $B$ are symmetric positive semidefinite matrices and $V=BB^T$. Then SDDE systems (\ref{SDDE1}) and (\ref{SDDE3}) have the same sample paths.
\end{theorem}
\begin{proof}
We need to show that if a given Wiener trajectory  $\textbf{W}(t)$ results in the sample path solution $\textbf{Y}(t)$ to (\ref{SDDE1}), there exists a Wiener trajectory $\textbf{W}^{\ast}(t)$ with the same sample path solution $\textbf{Y}^{\ast}(t)=\textbf{Y}(t)$ of (\ref{SDDE3}), and vice versa. Let us
assume that for a given Wiener trajectory $\textbf{W}(t)$ for $0\leq t\leq T$, SDDE (\ref{SDDE1}) has the sample path solution $\textbf{X}(t)$. Consider the following singular value decomposition of the matrix $G$:
\begin{equation*}
G(\textbf{Y}(t),\textbf{Y}(t-\tau_1),\dots,\textbf{Y}(t-\tau_q), t)=G(t)=P(t)C(t)Q(t),
\end{equation*}
for $0\leq t\leq T$, where $P(t)$ and $Q(t)$ are $n\times n$ and $m\times m$ orthogonal matrices, and $C(t)$ is a $n\times m$ matrix with $p\leq n$ positive diagonal entries. In light of orthogonality of matrices $P(t)$ and $Q(t)$, we have
$$V(t)=G(t)G(t)^T=P(t)C(t)C(t)^TP(t)^T=[B(t)]^2,$$
with
\begin{equation}\label{Beq}
B(t)=P(t)\big(C(t)C(t)^T\big)^{\sfrac{1}{2}}P(t)^T.
\end{equation}
One can then define the Wiener trajectory $\textbf{W}^{\ast}(t)$ for $0\leq t\leq T$ as
\begin{equation*}
\textbf{W}^{\ast}(t)=\displaystyle\int\limits_0^tP(s)\Big[\big(C(s)C(s)^T\big)^{\sfrac{1}{2}}\Big]^+C(s)Q(s)d\textbf{W}(s)+\displaystyle\int\limits_0^tP(s)d\textbf{W}^{\ast\ast}(s),
\end{equation*}
where $\textbf{W}^{\ast\ast}(s)$ is a vector of length $n$, whose first $p$ entries are equal to zero, and the remaining $n-p$ entries are independent Wiener processes, and $(\cdot)^+$ denotes the pseudo-inverse of a matrix \cite{watkins,ortega}. It follows that $\mathbb{E}(\textbf{W}^{\ast}(t)(\textbf{W}^{\ast}(t))^T)=tI_n$, where $I_n$ is the $n\times n$ identity matrix, thus confirming that $\textbf{W}^{\ast}(t)$ is indeed a vector of $n$ independent Wiener processes. Substituting $\textbf{Y}(t)$ instead of $\textbf{Y}^{\ast}(t)$ into the diffusion term of (\ref{SDDE3}) gives
\[
B(\textbf{Y}(t),\textbf{Y}(t-\tau_1),\dots,\textbf{Y}(t-\tau_q),t)d\textbf{W}^{\ast}(t)=B(t)\bigg[P(t)\Big[\big(C(t)C(t)^T\big)^{\sfrac{1}{2}}\Big]^+C(t)Q(t)d\textbf{W}(t)+P(t)d\textbf{W}^{\ast\ast}(t)\bigg],
\]
and using an expression for $B(t)$ from (\ref{Beq}), we obtain 
\[
B(\textbf{Y}(t),\textbf{Y}(t-\tau_1),\dots,\textbf{Y}(t-\tau_q),t)d\textbf{W}^{\ast}(t)=G\left(\textbf{Y}(t),\textbf{Y}(t-\tau_1),\dots,\textbf{Y}(t-\tau_q),t\right)d\textbf{W}(t),
\]
which proves that $\textbf{Y}(t)$ is a sample path solution of (\ref{SDDE3}).

Conversely, assume that a Wiener trajectory $\textbf{W}^{\ast}(t)$ for $0\leq t\leq T$ with the sample path solution $\textbf{Y}^{\ast}(t)$ to (\ref{SDDE3}) is given. Consider the following singular value decomposition
\begin{equation*}
B(\textbf{Y}^{\ast}(t),\textbf{Y}^{\ast}(t-\tau_1),\dots,\textbf{Y}^{\ast}(t-\tau_q), t)=B(t)=P(t)C(t)Q(t),
\end{equation*}
for $0\leq t\leq T$, where $P(t)$ and $Q(t)$ are $n\times n$ and $m\times m$ orthogonal matrices, respectively, with $C(t)$ being again an $n\times m$ matrix with $p\leq n$ positive diagonal entries. We can define the Wiener trajectory $\textbf{W}(t)$ for $0\leq t\leq T$ as
\begin{equation}\label{Wdef}
\textbf{W}(t)=\displaystyle\int\limits_0^tQ(s)^TC(s)^+\big[C(s)C(s)^T\big]^{\sfrac{1}{2}}P(s)^Td\textbf{W}^{\ast}(s)+\displaystyle\int\limits_0^tQ(s)^Td\textbf{W}^{\ast\ast\ast}(s),
\end{equation}
where $\textbf{W}^{\ast\ast\ast}(s)$ is a vector of length $m$, having zeros as the first $p$ entries, and the next $m-p$ entries being independent Wiener processes. Proceeding in the same way as above, we can show that the solution $\textbf{X}^{\ast}(t)$ of SDDE (\ref{SDDE3}) that corresponds to the Weiner trajectory $\textbf{W}^{\ast}(t)$ is also a solution of the SDDE (\ref{SDDE1}) corresponding to the Weiner trajectory $\textbf{W}(t)$ given in (\ref{Wdef}). Therefore, solutions to SDDE systems (\ref{SDDE1}) and (\ref{SDDE3}) have the same sample paths.
\end{proof}

Taken together, {\bf Theorems~\ref{NewThm}} and {\bf\ref{NewThm2}} show that any SDDE of the form
\begin{equation}\label{gen_model}
d\textbf{Y}=\textbf{f}\left(\textbf{Y}(t),\textbf{Y}(t-\tau_{m_1+1}),\dots,\textbf{Y}(t-\tau_m)\right)dt+Qd\textbf{W}(t),
\end{equation}
is equivalent to model (\ref{SDDE model}), as long as $QQ^T=V$ ($=HH^T$). This includes as a particular case system (\ref{SDDE2}) having a square matrix $Q$, but this does not necessarily have to be the case, provide the condition $QQ^T=HH^T$ is satisfied. The importance of this result is that since normally there is a large number of reactions involved, by allowing one to replace an $N\times (2m-m_2)$ matrix $H$ by a matrix with possibly much fewer columns, this can significantly reduce computational complexity of the resulting SDDE model.

\renewcommand{\arraystretch}{1.2}
\begin{table}
	\centering
	\caption{State changes $\Delta\textbf{Y}$ in a small time interval $\Delta t$}
	\begin{tabular}{|l|l|l|}
		\hline 
		$i$      & $(\Delta\textbf{Y})_i$  & Probability $P_i\Delta t$                       \\ \hline
		1        & $\textbf{v}_1$          & $a_1(\textbf{Y}(t))\Delta t$                    \\ \hline
		\vdots   & \vdots                  & \vdots                                          \\ \hline
		$m_1$    & $\textbf{v}_{m_1}$      & $a_{m_1}(\textbf{Y}(t))\Delta t$                \\ \hline
		$m_1+1$  & $\textbf{v}_{m_1+1}$    & $a_{m_1+1}(\textbf{Y}(t-\tau_{m_1+1}))\Delta t$ \\ \hline
		\vdots   & \vdots                  & \vdots                                          \\ \hline
		$m_2$    & $\textbf{v}_{m_2}$      & $a_{m_2}(\textbf{Y}(t-\tau_{m_2}))\Delta t$     \\ \hline
		$m_2+1$  & $\textbf{v}_{m_2+1}^r$  & $a_{m_2+1}(\textbf{Y})\Delta t$              \\ \hline
		\vdots   & \vdots                  & \vdots                                          \\ \hline
		$m$      & $\textbf{v}_{m}^r$      & $a_{m}(\textbf{Y})\Delta t$                  \\ \hline
		$m+1$    & $\textbf{v}_{m_2+1}^p$  & $a_{m_2+1}(\textbf{Y}(t-\tau_{m_2+1}))\Delta t$ \\ \hline
		\vdots   & \vdots                  & \vdots                                          \\ \hline
		$2m-m_2$ & $\textbf{v}_{m}^p$      & $a_{m}(\textbf{Y}(t-\tau_{m}))\Delta t$         \\ \hline
		$2m-m_2+1$ & \textbf{0}            & $1-\sum\limits_{i=1}^{2m-m_2} P_i\Delta t$      \\ \hline
	\end{tabular}
	\label{possible changes1}
\end{table}
\renewcommand{\arraystretch}{1}

\section{Algorithm for deriving an SDDE}

Having established the equivalence of systems (\ref{SDDE model}) and (\ref{gen_model}), let us present an alternative approach for finding the function $\textbf{f}$ and the matrix $Q$, which extends the method presented earlier in \cite{allen2007b,Fatehi2018a} for systems without time delays. Let us recall that $\textbf{Y}(t)=(Y_1(t),Y_2(t),\ldots,Y_N(t))^T$ is a vector of continuous random variables representing the amounts of molecular species at time $t$, with the first $m_1$ reactions being non-delayed, reactions $m_1+1$ to $m_2$ being delayed non-consuming reactions with corresponding time delays $\tau_{m_1+1},\ldots,\tau_{m_2}$, and the rest to be consuming delayed reactions with time delays $\tau_{m_2+1},\ldots,\tau_{m}$. We assume that $\Delta t$ is small enough, so that during this time interval at most one change can occur in state variables as represented by the state change vectors, and if it is a consuming delay reaction, then we split its state change vector into two vectors in a similar way to how it was done for the DCME (\ref{DCME1}), namely, with one state change vector representing changes in reactants, and the second one representing changes in products. These state changes together with corresponding probabilities are all listed in Table~\ref{possible changes1}. Using this table of possible state changes, one can compute the expectation and covariance matrix of the state change $\Delta\textbf{Y}$ for sufficiently small $\Delta t$.

The expectation vector to order $\Delta t$ is given by
\begin{equation*}
\mathbb{E}(\Delta \textbf{Y})\approx\sum\limits_{i=1}^{2m-m_2}P_i(\Delta\textbf{Y})_i\Delta t=\boldsymbol{\mu}\Delta t,
\end{equation*}
and the covariance matrix is obtained by only keeping terms of order $\Delta t$, i.e.
\begin{align*}
\mbox{cov}(\Delta\textbf{Y})&=\mathbb{E}\left[(\Delta \textbf{Y})(\Delta \textbf{Y})^T\right]-\mathbb{E}\left[\Delta \textbf{Y}\right](\mathbb{E}\left[\Delta \textbf{Y}\right])^T\approx \mathbb{E}\left[(\Delta \textbf{Y})(\Delta \textbf{Y})^T\right]\\
&=\sum\limits_{i=1}^{2m-m_2}P_i(\Delta \textbf{Y})_i(\Delta \textbf{Y}_i)^T\Delta t=\Sigma \Delta t.
\end{align*}
It can be easily shown that for the matrix $H$ in equation (\ref{SDDE model}) $HH^T=\Sigma$, and thus, matrix $\Sigma$ found using Table~\ref{possible changes1} is the same as matrix $V$.

In summary, to derive an SDDE model for a model with delays $\{\tau_1,\tau_2,\allowbreak\ldots,\tau_q\}$, first we have to compile the table with all possible state changes, explicitly separating consuming reactions. Then we use this table to find the drift vector $\boldsymbol{\mu}$ and covariance (diffusion) matrix $\Sigma$, from which we find the matrix $Q$ satisfying $QQ^T=\Sigma$. The resulting It\^{o} SDDE model then has the form
\begin{equation}\label{k}
\begin{cases}
d\textbf{Y}(t)=\boldsymbol{\mu}dt+Qd\textbf{W}(t), \\
\textbf{Y}(t)=\boldsymbol{\varphi}(t)\quad \mbox{for}\quad t\in [-\tau,0],
\end{cases}
\end{equation}
where $\tau=\max\{\tau_1,\dots,\tau_q\}$, and $\textbf{W}(t)$ is a vector of independent Wiener processes. One should note that the order of entries in the table of state changes is irrelevant, since all entries come with their respective probabilities. Moreover, if any two (or more) entries have the same  state change vectors, these entries can be combined into one, with the associated probability being the sum of individual probabilities of those entries. This would reduce the size of the tables of state changes, but would not affect the drift vector or the diffusion matrix.\\

\noindent {\bf Remark 1.} {\it In order for SDDE model (\ref{SDDE model}), which represents a delayed chemical Langevin equation (CLE), to provide a good approximation of the original DCME (\ref{DCME1}), certain assumptions have to be satisfied. The first of these is the so-called {\it leap condition} \cite{gillespie2000}, which states that there exists some $\Delta t>0$, such that propensities for all reactions $a_j(\bf X)$ remain constant on time interval $[t,t+\Delta t)$. This then implies that the number of reactions $R_j$ that occur in the interval $[t,t+\Delta t)$ obeys a Poisson distribution with parameter $a_j(\bf x)\Delta t$, where ${\bf X}(t)={\bf x}$. Under additional assumption that $\Delta t$ is not only small enough to satisfy the leap condition, but also large enough to satisfy $a_j(\bf x)\Delta t\gg 1$, one can approximate each Poisson random variable for a normal random variable with the same mean and variance, $\displaystyle{P_j(a_j(\bf x)\Delta t)\approx \mathcal{N}_j\left(a_j(\bf x)\Delta t,a_j(\bf x)\Delta t\right)=a_j(\bf x)\Delta t+\mathcal{N}_j(0,1)}$ \cite{gillespie2000}. Both of these conditions are satisfied when the numbers of species involved are large \cite{ilie09,grima11}, but this is a sufficient condition, and Grima et al. \cite{grima11} have shown that in certain regimes even for relatively small numbers of species, CLE can still provide a good approximation of the CME. An alternative derivation of the CLE can be found in M\'elyk\'uti et al. \cite{mely10}, where it was shown that CLEs form a parametric family of equivalent equations. In the case of delayed CLE, there is an additional assumption $P({\bf X}-{\bf v}_j,t; {\bf X}_i,t-\tau_j)\approx P({\bf X}-{\bf v}_j,t)\times P({\bf X}_i,t-\tau_j)$, which effectively means that the time delays are sufficiently large to ensure that a larger number of reactions occur during a time interval $[t-\tau_j,t]$, so that the coupling of system states at time $t-\tau_j$ and $t$ is weak \cite{tian2007}.}

\section{Examples}

To illustrate how the methodology developed in the previous section can be used for deriving and simulating stochastic models with consuming and non-consuming delayed reactions, below we consider two specific examples, where the mean-field deterministic analogues are characterised either by a single stable steady state, or by a bi-stability between two stable steady states.

\subsection{Example 1}

Let us consider a system of two molecular species, whose state at time $t$ is described by the vector $\textbf{X}(t)=(X(t),Y(t))$, which interact through the following set of reactions
\begin{equation}
\begin{array}{l}
\displaystyle{R_1: \varnothing \xrightarrow{b} X,\quad R_2: X\xrightarrow{d}\varnothing, \quad R_3: Y\xrightarrow{d}\varnothing,\quad
R_4:X+Y\xdashrightarrow{a,\tau_4}2Y,\quad R_5: Y\xdashrightarrow{c,\tau_5} X,}
\end{array}
\end{equation} 
where instantaneous reactions are indicated with solid arrows, and time-delayed reactions are shown with dashed arrows, with all reaction rates shown above the corresponding arrows. This system has three non-delayed reactions $R_1$ to $R_3$, reaction $R_4$ is a non-consuming delayed reaction, and reaction $R_5$ is a consuming delayed reaction. Using the notation from the previous section, we introduce $m_1=3$, $m_2=4$, and $m=5$ as an overall number of reactions.

Using the law of mass action, we obtain a system of differential equations describing deterministic evolution of mean-field concentrations of species $X$ and $Y$
\begin{equation}\label{DDE model}
	\begin{array}{l}
		\displaystyle{\dfrac{dX}{dt}=b-dX(t)-{a}X(t-\tau_4)Y(t-\tau_4)+cY(t-\tau_5),}\\\\
		\displaystyle{\dfrac{dY}{dt}={a}X(t-\tau_4)Y(t-\tau_4)-cY(t)-dY(t).}
	\end{array}
\end{equation}
This model can have up to two steady states: $E_1=(b/d,0)$ and $E_2=((c+d)/a,(ab-d(c+d))/ad)$. $E_1$ is stable for any $\tau_4$ and $\tau_5$ if $ab<d(c+d)$, and unstable for any time delays if $ab>d(c+d)$, in which case the second steady state is feasible, i.e. both of its components are positive.

To derive an SDDE model, we consider $\textbf{Y}(t)=(Y_1(t), Y_2(t))$ to be a vector of continuous random variables describing the amounts of species $X$ and $Y$ and time $t$. Following the method described in the previous section, we conclude that there are $2m-m_2=6$ state changes that have to be included. Under assumption of $\Delta t$ being sufficiently small to ensure that during this time interval at most one change can occur in state variables, these state changes together with their probabilities are shown in Table~\ref{possible changes}. 
\renewcommand{\arraystretch}{1.6}
\begin{table}[H]
	\centering
	\caption{Possible state changes $\Delta\textbf{Y}$ during a small time interval $\Delta t$.}
	\begin{tabular}{|l|l|l|}
		\hline
		$i$& $(\Delta\textbf{Y})_i^T$ & Probability $P_i\Delta t$ \\ \hline
		1  & $(1,0)$  & $b\Delta t$ \\ \hline
		2  & $(-1,0)$  & $dY_1(t)\Delta t$ \\ \hline
		3  & $(0,-1)$ & $dY_2(t)\Delta t$ \\ \hline
		4  & $(-1,1)$ & $aY_1(t-\tau_4)Y_2(t-\tau_4)\Delta t$ \\ \hline
		5  & $(0,-1)$ & $cY_2(t)\Delta t$ \\ \hline
		6  & $(1,0)$ & $cY_2(t-\tau_5)\Delta t$ \\ \hline
		7  & $(0,0)$  & $1-\sum\limits_{i=1}^{6} P_i\Delta t$ \\ \hline
	\end{tabular}
	\label{possible changes}
\end{table}
\noindent Using Table~\ref{possible changes}, the expectation vector $\mathbb{E}(\Delta \textbf{Y})$ and covariance matrix $\mbox{cov}(\Delta\textbf{Y})$ to order $\Delta t$ can be found as
\[
\displaystyle{\mathbb{E}(\Delta \textbf{Y})\approx\sum\limits_{i=1}^{6}P_i(\Delta\textbf{Y})_i\Delta t=\boldsymbol{\mu}\Delta t\quad 
\Longrightarrow\quad \boldsymbol{\mu}=\begin{pmatrix}
P_1-P_2-P_4+P_6 \\
P_4-P_3-P_5 
\end{pmatrix},}
\]
and
\[
\displaystyle{\mbox{cov}(\Delta\textbf{Y})\approx\sum\limits_{i=1}^{6}P_i(\Delta \textbf{Y})_i(\Delta \textbf{Y}_i)^T\Delta t=\Sigma \Delta t\quad 
\Longrightarrow\quad \Sigma=\begin{pmatrix}
	P_1+P_2+P_4+P_6 & -P_4 \\
	-P_4 & P_3+P_4+P_5 
	\end{pmatrix}.}
\]
\begin{figure}
	\centering
	\includegraphics[width=1\linewidth]{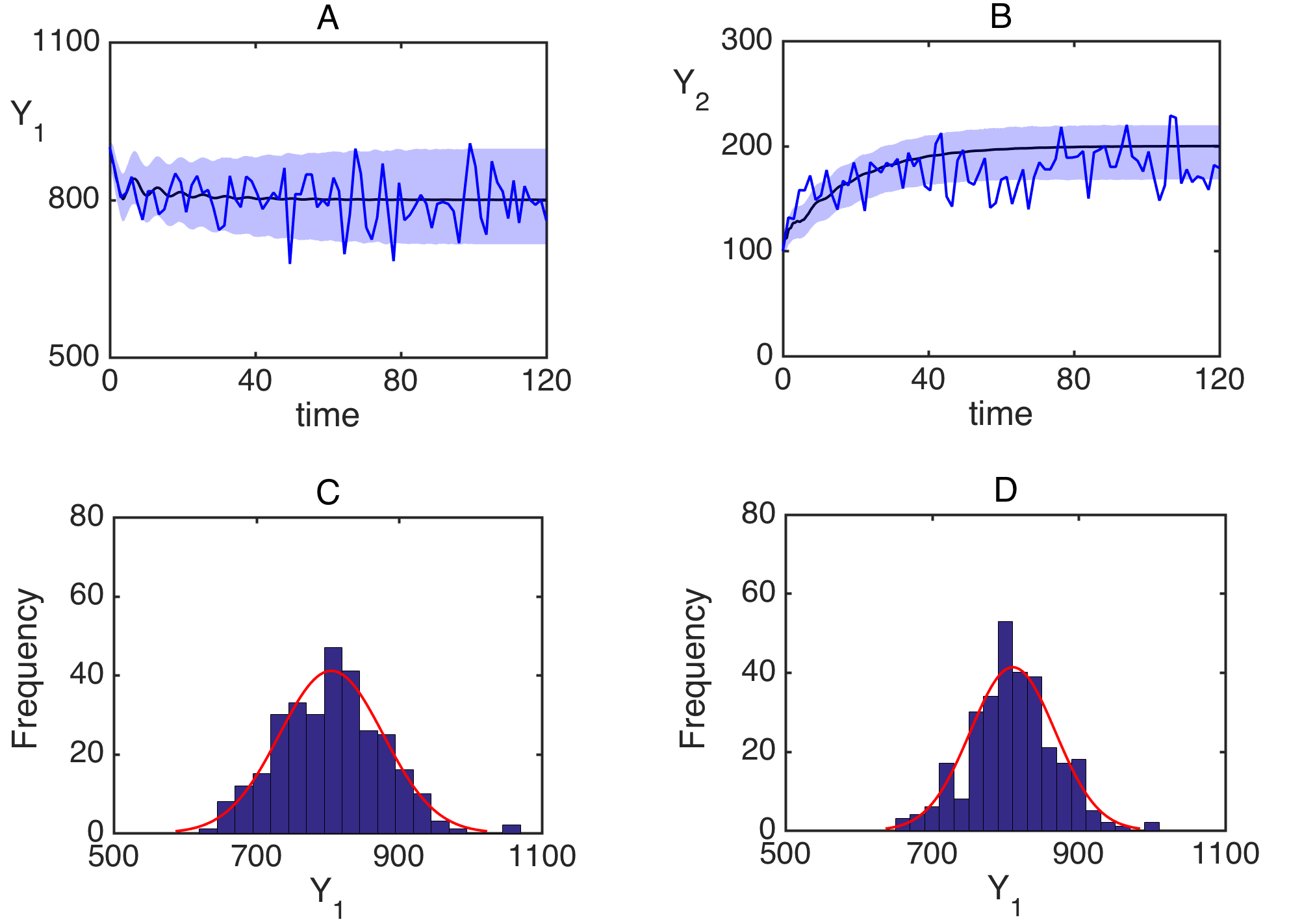}
	\caption{\small (a), (b) Numerical simulation of the SDDE (\ref{SDDE}), where shaded blue region indicates an area of one standard deviation from the mean of 10000 simulations. Blue curve shows one stochastic realisation of the model (\ref{SDDE}), black curve is the solution of the deterministic model (\ref{DDE model}). (c) and (d) show frequency distributions at $t=120$ of values for the variable $Y_1$ using SDDE and DSSA with 300 simulations, respectively, together with a fit to a normal distribution shown in red. Parameter values are $a=0.005$, $b=1000$, $c=1$, $d=1$, $\tau_4=1$, $\tau_5=3.5$.}
	\label{meanstd}
\end{figure}
If we now define the matrix $Q$ as follows,
\[
Q=\begin{pmatrix}
\sqrt{P_1+P_2+P_6} & -\sqrt{P_4} & 0 \\
0 & \sqrt{P_4} & \sqrt{P_3+P_5}
\end{pmatrix},
\]
then the $2\times 3$ matrix $Q$ satisfies $QQ^T=\Sigma$, and the It\^{o} SDDE model thus has the form
\begin{equation}\label{SDDE}
\begin{cases}
d\textbf{Y}(t)=\boldsymbol{\mu}dt+Qd\textbf{W}(t), \\
\textbf{Y}(t)=\boldsymbol{\varphi}(t)\quad \mbox{for}\quad t\in [-\tau,0],
\end{cases}
\end{equation}
where $\textbf{W}(t)=(W_1(t),W_2(t),W_3(t))^T$ is a vector of three independent Wiener processes, $\tau=\max\{\tau_4,\tau_5\}$, and $\boldsymbol{\varphi}(t)$ is the vector of initial conditions. It is noteworthy that the matrix $Q$ is only $2\times 3$, and not $2\times 6$ as it would be in the original SDDE formulation (\ref{SDDE model}), thus reducing the number of independent Wiener processes required for computation by half.

To solve the model (\ref{SDDE}) numerically, we use the strong predictor-corrector method with the degree of implicitness in the drift coefficient chosen to be equal to $1/7$, since for this value the method has the largest stability region \cite{niu2014,niu2015}. We choose the values of parameters in such a way that the steady state $E_2$ is feasible and deterministically stable. The initial condition is taken to be
\begin{equation}\label{IC}
\left(Y_1(s),Y_2(s)\right)=(900,100), \quad s\in[-\tau_{max},0],\quad \tau_{max}=\max\{\tau_4,\tau_5\}.
\end{equation}
Figure~\ref{meanstd} shows the results of numerical solution of the model (\ref{SDDE}) with initial conditions (\ref{IC}) for 10,000 realisations. Since deterministically the steady state $E_2$ is stable (and the system is in its basin of attraction), solution of the deterministic model (\ref{DDE model}) approaches this steady state, while initially exhibiting some decaying oscillations associated with characteristic eigenvalues of $E_2$ being complex and having a small negative real part. Stochastically, the mean is very close to the deterministic trajectory, because it obeys the same deterministic system of equations \cite{Allen2010}, and we also observe that as time progresses, the variance of stochastic solutions settles on some steady level. One can also notice that even though averaged dynamics mimic the behaviour of the deterministic model, individual stochastic realisations exhibit sustained oscillations, a phenomenon known as coherence resonance or {\it stochastic amplification} \cite{alonso2007,Kuske2007}.

Although it is known that SDDEs only provide an approximation of the true stochastic dynamics, in Figs.~\ref{meanstd}(c) and (d) we have compared the distribution of values for one of the species obtained as a solution of the SDDE model (\ref{SDDE}) with an equivalent distribution obtained using an exact DSSA proposed by Cai \cite{Cai2007} and implemented in the StochPy package in Python \cite{stochpy13}. One observes a good agreement between the two distributions, providing additional support for using SDDEs as an effective tool for stochastic simulations of systems with consuming and non-consuming delays. Importantly, with both SDDE and DSSA codes being implemented in Python and run on the same laptop with 2.6GHz i7-3720 processor, one run of the SDDE model only took on average 0.1 sec, while one run of the DSSA took on average 55.2 sec, suggesting a huge improvement in terms of speed of performance, without compromising accuracy in terms of resulting distribution.

\subsection{Example 2}

It has been extensively discussed in the context of various biological and chemical models that negative feedback is required for systems to exhibit oscillations, while positive feedback is needed for multi-stability, see, e.g. \cite{fer02,krishna09,wilh09} and references therein. As our second example, we consider a model suggested by Wilhelm \cite {wilh09} with positive and negative feedback, which arguably represents the smallest bistable chemical reaction system in terms of having the smallest numbers of reactants, reactions, and associated ODEs representing chemical kinetics. This model consists of two species $X$ and $Y$ that interact as shown in the diagram below
\begin{figure}[H]
	\centering
	\includegraphics[width=0.3\linewidth]{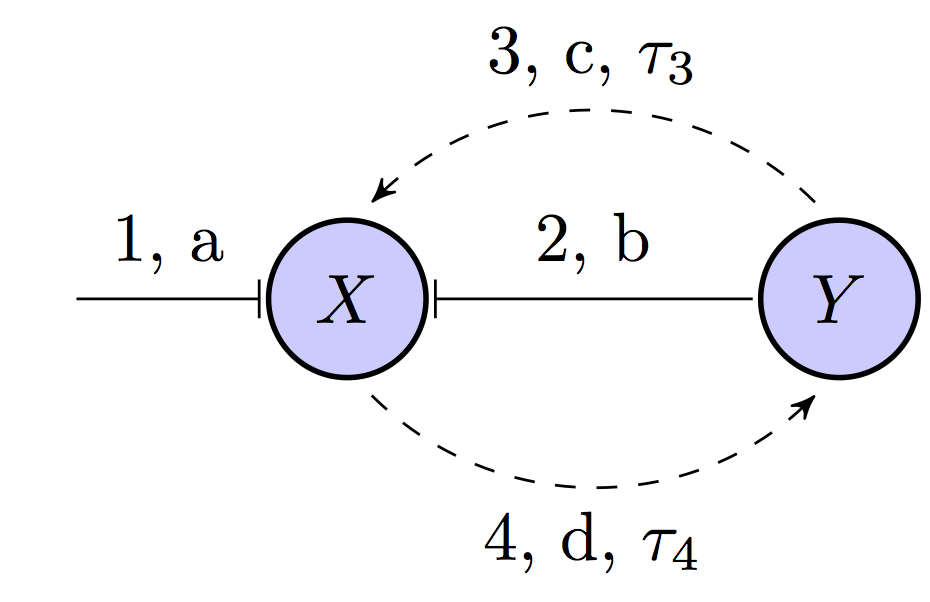}
\end{figure}
\noindent which corresponds to the following systems of reactions
\begin{equation}\label{ex2_syst}
\begin{array}{l}
\displaystyle{R_1: X \xrightarrow{a} \varnothing,\quad R_2: X+Y\xrightarrow{b} Y, \quad R_3: Y\xdashrightarrow{c,\tau_3} 2X,\quad
R_4:2X\xdashrightarrow{d,\tau_4}Y,}
\end{array}
\end{equation} 
where, as before, solid arrows represent instantaneous reactions, and dashed lines represent time-delayed reactions. We also assume that reactions $R_1$ and $R_2$ are non-delayed, reaction $R_3$ is a non-consuming delayed reaction, and $R_4$ is a consuming delayed reaction. Using the notation from the pervious section, this gives $m_1=2$, $m_2=3$, and $m=4$.

Applying the law of mass, one obtains a system of two ODEs describing the dynamics of mean-field concentrations of chemical species $X$ and $Y$:
\begin{equation}\label{ex2_det}
\begin{array}{l}
\displaystyle{\frac{dX}{dt}=2cY(t-\tau_3)-dX(t)^2-bX(t)Y(t)-aX(t),}\\\\
\displaystyle{\frac{dY}{dt}=dX(t-\tau_4)^2-cY(t-\tau_3).}
\end{array}
\end{equation}
For any values of parameters, this system has a trivial steady state $(X_0,Y_0)=(0,0)$, and provided $cd>4ab$, it also has a pair of additional steady states
\[
\displaystyle{(X^*,Y^*)=\frac{c}{2b}\left(1\pm\sqrt{1-4\frac{ab}{cd}}\right).}
\]

In order to derive an SDDE representation of the model, we introduce a vector ${\bf Y}=(Y_1(t),Y_2(t))$ of continuous random variables, whose components represent the amounts of chemical species $X$ and $Y$ at any given time $t$. In this case, there are $2m-m_2=5$ different transitions to consider during any infinitesimal time interval $\Delta t$, and their probabilities, as well as associated state changes, are given in the following table
\renewcommand{\arraystretch}{1.6}
\begin{table}[H]
	\centering
	\caption{Possible state changes $\Delta\textbf{Y}$ during a small time interval $\Delta t$.}
	\begin{tabular}{|l|l|l|}
		\hline
		$i$& $(\Delta\textbf{Y})_i^T$ & Probability $P_i\Delta t$ \\ \hline
		1  & $(-1,0)$  & $aY_1(t)\Delta t$ \\ \hline
		2  & $(-1,0)$  & $bY_1(t)Y_2(t)\Delta t$ \\ \hline
		3  & $(2,-1)$ & $cY_2(t-\tau_3)\Delta t$ \\ \hline
		4  & $(-1,0)$ & $dY_1(t)^2\Delta t$ \\ \hline
		5  & $(0,1)$ & $dY_1(t-\tau_4)^2\Delta t$ \\ \hline
		6  & $(0,0)$  & $1-\sum\limits_{i=1}^{5} P_i\Delta t$ \\ \hline
	\end{tabular}
	\label{Ex2_table}
\end{table}
\setcounter{figure}{1}
\begin{figure}
	\centering
	\includegraphics[width=1\linewidth]{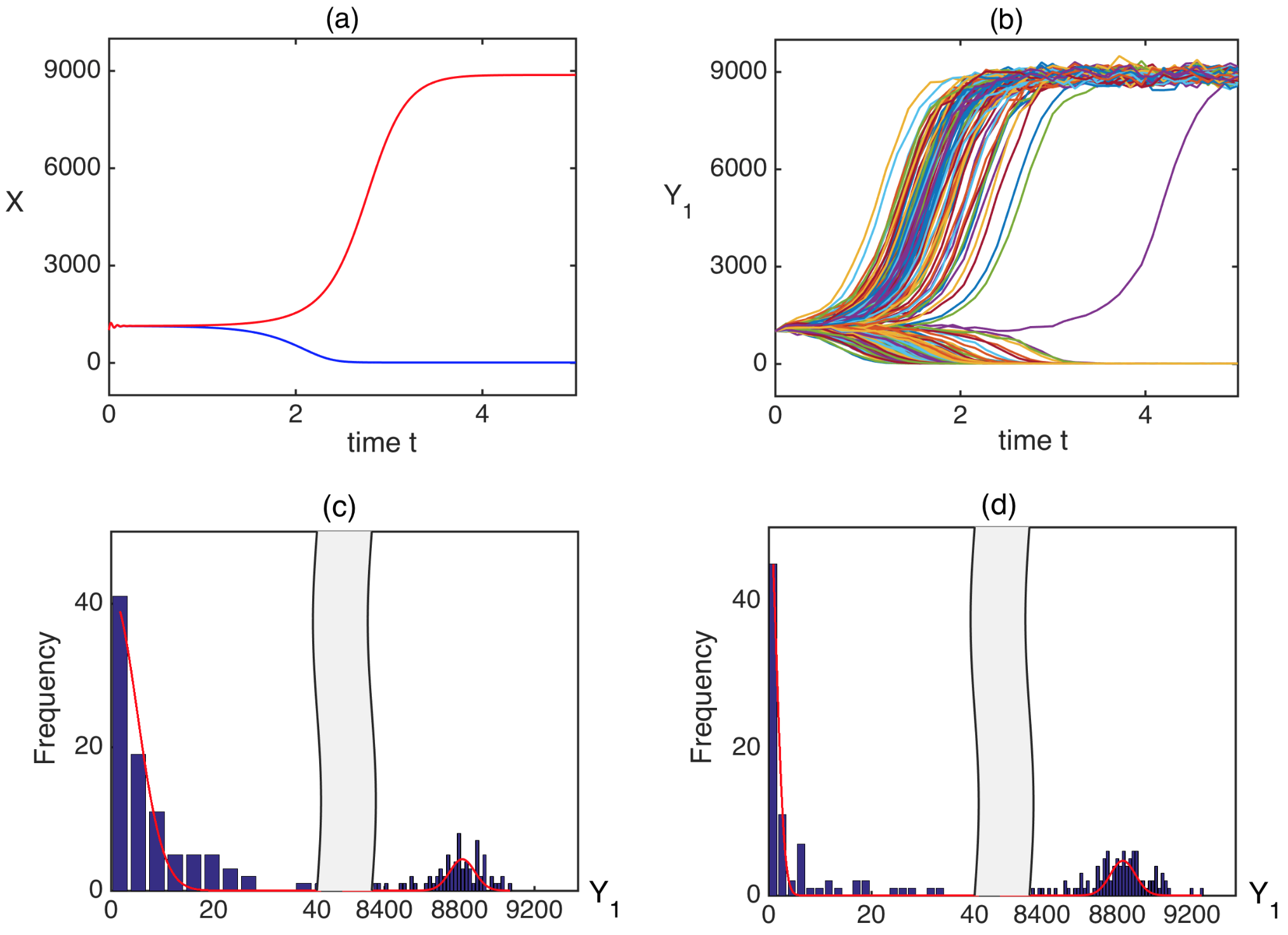}
	\caption{\small (a) Numerical solution of the deterministic model (\ref{ex2_det}) with initial conditions $(X_0(s),Y_0(s))=(1000,395)$ (red) and $(X_0(s),Y_0(s))=(1000,392)$ (blue) on $s\in [-0.03,0]$. (b) Numerical simulation of the SDDE (\ref{ex2_SDDE}) with initial condition $(X_0(s),Y_0(s))=(1000,392)$ on $s\in [-0.03,0]$ for 200 realisations. (c) and (d) show frequency distributions at $t=3$ of values for the variable $Y_1$ using SDDE and DSSA, respectively, together with a fit to a bi-normal distribution shown in red. Parameter values are $a=10$, $b=0.005$, $c=50$, $d=0.01$, $\tau_3=0.01$, $\tau_4=0.03$.}
	\label{bist_model}
\end{figure}
This immediately gives approximations to order $\Delta t$ of the expectation vector 
\[
\displaystyle{\mathbb{E}(\Delta \textbf{Y})\approx\sum\limits_{i=1}^{5}P_i(\Delta\textbf{Y})_i\Delta t=\boldsymbol{\mu}\Delta t\quad 
\Longrightarrow\quad \boldsymbol{\mu}=\begin{pmatrix}
-P_1-P_2+2P_3-P_4 \\
-P_3+P_5 
\end{pmatrix},}
\]
and the covariance matrix
\[
\displaystyle{\mbox{cov}(\Delta\textbf{Y})\approx\sum\limits_{i=1}^{6}P_i(\Delta \textbf{Y})_i(\Delta \textbf{Y}_i)^T\Delta t=\Sigma \Delta t\quad 
\Longrightarrow\quad \Sigma=\begin{pmatrix}
	P_1+P_2+4P_3+P_4 & -2P_3 \\
	-2P_3 & P_3+P_5 
	\end{pmatrix}.}
\]
Introducing the matrix $Q$ as follows,
\[
Q=\begin{pmatrix}
\sqrt{P_1+P_2+P_4} & 2\sqrt{P_3} & 0 \\
0 & -\sqrt{P_3} & \sqrt{P_5}
\end{pmatrix},
\]
ensures that it satisfies the condition $QQ^T=\Sigma$, and therefore, the It\^{o} SDDE model for system (\ref{ex2_syst}) has the form
\begin{equation}\label{ex2_SDDE}
\begin{cases}
d\textbf{Y}(t)=\boldsymbol{\mu}dt+Qd\textbf{W}(t), \\
\textbf{Y}(t)=\boldsymbol{\varphi}(t)\quad \mbox{for}\quad t\in [-\tau,0],
\end{cases}
\end{equation}
with $\textbf{W}(t)=(W_1(t),W_2(t),W_3(t))^T$ being a vector of three independent Wiener processes, $\tau=\max\{\tau_3,\tau_4\}$, and $\boldsymbol{\varphi}(t)$ being the vector of initial conditions. Similarly to the first example, we have reduced the number of independent Wiener processes required for computation from 5 to 3.

For each particular choice of time delays $\tau_3$ and $\tau_4$, deterministic model (\ref{ex2_det}) exhibits a bistability, where for the same values of parameters, the solution approaches either a trivial steady state $(0,0)$, or a non-trivial equilibrium $(X^*,Y^*)$, depending on the initial condition. Figure~\ref{bist_model}(a) illustrates such behaviour, where for a very small difference in the initial values of $Y$ variable, the solution with higher initial $Y$ goes to $(X^*,Y^*)$, while the solution with smaller initial $Y$ approaches a steady state $(0,0)$. For the same values of parameters and time delays, if we choose initial condition that deterministically approaches the steady state $(0,0)$, in the case of SDDE model (\ref{ex2_SDDE}) we observe that solutions will approach either of the two steady states with some probability, as shown in Fig.~\ref{bist_model}(b). Comparing distribution of frequencies with an equivalent distribution obtained by solving the original model using a delayed next reaction method \cite{anderson07}, which is another exact DSSA implemented in StochPy package \cite{stochpy13}, we again observe good qualitative agreement, while having a very substantial decrease in computational time.

\begin{figure}
	\centering
	\includegraphics[width=1\linewidth]{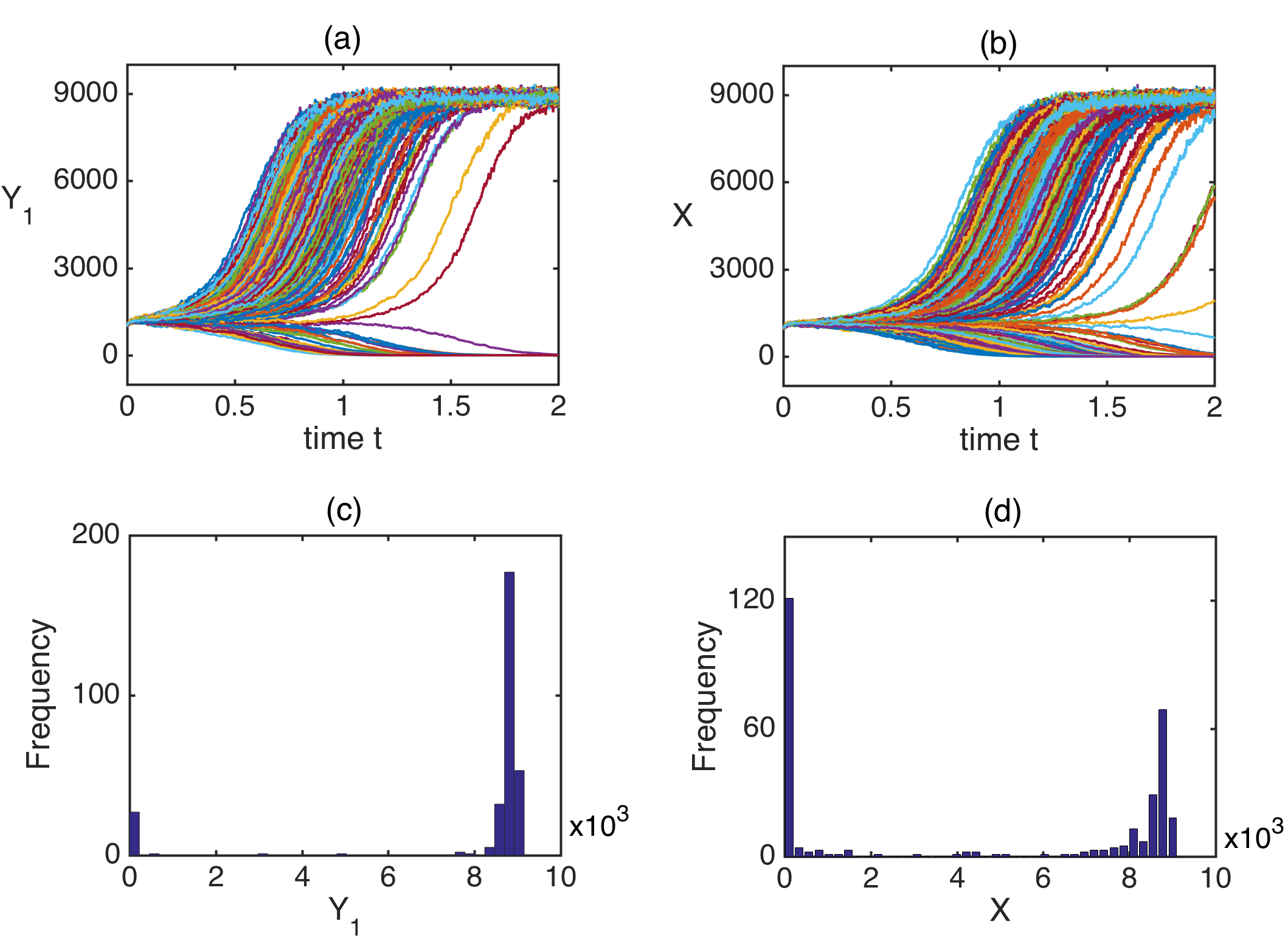}
	\caption{\small (a) Numerical solution of the SDDE (\ref{ex2_SDDE}) with initial condition $(X_0(s),Y_0(s))=(1000,392)$ on $s\in [-0.003,0]$ for 300 realisations. (b) Numerical solution of the model (\ref{ex2_syst}) using DSSA with initial condition $(X_0(s),Y_0(s))=(1000,392)$ on $s\in [-0.003,0]$ for 300 realisations. (c) and (d) show frequency distributions at $t=1.5$ of values for the variable $Y_1$ using SDDE and DSSA, respectively. Parameter values are $a=10$, $b=0.005$, $c=50$, $d=0.01$, $\tau_3=0.001$, $\tau_4=0.003$.}
	\label{error_plot}
\end{figure}

The validity of CLE approximation to CMEs has been earlier studied numerically in the context of non-delayed \cite{gillespie2000,grima11} and delayed \cite{tian2007} systems from the perspective of not very large system sizes. Going back to {\bf Remark 1}, we have looked into how the accuracy of this approximation is affected by sufficiently small delays, which can potentially violate one of the assumptions behind the derivation of the delayed CLE regarding weak coupling between system states at times $t-\tau_j$ and $t$. To investigate this issue, we have fixed the values of all parameters as in Fig.~\ref{bist_model}, keeping large numbers of species, but reduced both time delays by a factor of 10. Corresponding simulations, as shown in Fig.~\ref{error_plot}, suggest that whereas the aggregate differences between temporary profiles of solutions obtained using SDDE and the exact DSSA appear to be small, the details of those solutions are quite different. While only 10\% of solutions of SDDE model approached the trivial steady state at the end of simulation, this proportion rose to 45\% for solutions obtained using the DSSA. Furthermore, looking at distribution profiles, one observes a larger clustering of solutions close to $Y_1=0$ for the case where DSSA was used, as compared to a much higher peak around $Y_1=Y^*$ for the SDDE model. This suggests that while generally SDDE-based models of stochastic systems with delays can provide computationally efficient approximations of stochastic dynamics, in the case of very small time delays, the accuracy of the approximation provided by these models can be reduced, thus necessitating the use of DSSAs to simulate the dynamics.

\section{Discussion}

In this paper we have shown that a number of alternative formulations of SDDE models can be obtained that are all equivalent in terms of probability distribution and sample paths. Using this equivalence, we have proposed an algorithm for deriving computationally efficient It\^{o} SDDEs from the DCMEs for systems with consuming and non-consuming delayed reactions. Numerical simulations done on an example of a system with two chemical species interacting through five non-delayed, delayed non-consuming and delayed consuming reactions, show that the distributions obtained as solutions of such SDDEs provide a good approximation of the exact dynamics, but are significantly faster than delayed stochastic simulation algorithms. Similarly, good agreement was observed between the results of an SDDE formulation for a chemical reaction model with bistability and an exact solution computed using a DSSA. It is important to note, though, that SDDE models described in this paper can only provide accurate approximations of underlying stochastic dynamics in certain regimes, and the accuracy of this approximation can deteriorate for small delays, as was observed in the example with bistability.

In many scenarios, discrete time delays (which effectively are represented by $\delta$-functions) provide reasonable approximation for various biological processes that happen non-instantaneously. However, in some cases such description is not adequate, and it would be more appropriate to represent time delays by proper distributions \cite{laf13}. One example is stochastic models of epidemics, where distribution of infectious periods is much closer to a $\Gamma$-distribution, which interpolates between constant and exponentially distributed infectious periods \cite{lloyd00}. Representing such a distribution by a number of infectious stages, with individuals progressing through stages and staying in each stage for exponentially distributed periods of time, it is possible to derive a master equation describing the dynamics, from which a power spectrum of stochastic oscillations can be analytically obtained \cite{black09}. Using such an approach, known in other contexts as a `linear chain trick' \cite{mcdonald}, one effectively avoids the need for having a delayed distribution in the model, thus making the resulting system of SDEs much easier to solve numerically. A somewhat similar strategy, but in reverse, was proposed by Barrio et al. \cite{barrio13} to abridge large chains of consecutive reactions by lumping reaction sequences into smaller systems, which provided an improvement when using a stochastic simulation algorithm. However, in many realistic situations it does not prove possible to simplify the delay distribution, and in the future we will consider how one could generalise an approach presented in this paper for stochastic models with distributed delays.

\bibliographystyle{ieeetr}
\bibliography{timedelaystochastic}

\end{document}